\documentclass[a4paper]{scrartcl}

\makeatletter

\usepackage{geometry} 
\usepackage{amssymb}
\usepackage{palatino}
\usepackage{eulervm}

\usepackage[utf8]{inputenc}
\usepackage{amsmath}
\usepackage[amsmath,hyperref,thmmarks]{ntheorem}
\usepackage{mathrsfs}
\usepackage{tikz}
\usetikzlibrary{shapes,arrows,chains,shadows}
\usepackage{mathpartir} 

\usepackage{hyperref}

\renewcommand{\emph}[1]{{\color{blue}\textbf{#1}}}
\newcommand{\N}[1]{\ensuremath{{\color{blue}{#1}}}}
\newcommand{\M}[1]{\ensuremath{\mathrm{#1}}}

\usepackage[neveradjust,flushleft]{paralist}
\renewenvironment{enumerate}{\compactenum}{\endcompactenum}
\renewenvironment{itemize}{\compactitem}{\endcompactitem}
\newlength{\baseindent}
\setdefaultleftmargin{\baseindent}    
{\baseindent}{\baseindent}{\baseindent}
{\baseindent}{\baseindent}
\newcommand{\nomathmargin}{\setlength{\@mathmargin}{0em}}

\usepackage{coqtheorem}

\theoremstyle{coqtheorem}
\theoremheaderfont{\color{blue}\rmfamily\bfseries}
\theorembodyfont{\normalfont}
\newtheorem{lemma}{Lemma}
\newtheorem{fact}[lemma]{Fact}
\newtheorem{theorem}[lemma]{Theorem}
\newtheorem{corollary}[lemma]{Corollary}

\theoremstyle{nonumberplain}
\theorembodyfont{\normalfont}
\theoremsymbol{{\color{blue}\ensuremath{_\blacksquare}}}
\newtheorem{proof}{Proof}

\makeatother

\setBaseUrl{{https://www.ps.uni-saarland.de/extras/S1S/website/}}

\renewcommand{\emph}[1]{{\color{blue}\textbf{#1}}}
\renewcommand\N[1]{{\color{blue}#1}}
\renewcommand\M[1]{\ensuremath{\mathsf{#1}}}
\renewcommand{\nvDash}{\not\models}
\newcommand{\bigand}{\bigwedge}

\newcommand{\ol}[1]{\overline{#1}}
\newcommand{\lam}[2]{\lambda#1{.}\hskip.7pt#2}
\newcommand{\toot}{\leftrightarrow}
\newcommand{\set}[1]{\ensuremath{\{#1\}}}
\newcommand{\mset}[2]{\set{\,#1\mid#2\,}}

\newcommand{\incl}{\ensuremath{\subseteq}}
\newcommand\nat{\M{N}}
\newcommand\oneM{\oldstylenums{1}}
\newcommand\one{\M{1}}
\newcommand\bool{\M{2}}
\newcommand{\btrue}{\mathsf{true}}
\newcommand{\bfalse}{\mathsf{false}}
\renewcommand\phi\varphi
\newcommand\var{\M{V}}
\newcommand\xm{\M{xm}}
\newcommand\XM{\M{XM}}
\newcommand\FX{\M{FX}}
\newcommand\RF{\M{RF}}
\newcommand\RA{\M{RA}}
\newcommand\RP{\M{RP}}
\newcommand{\RPc}[0]{\RP^\mathsf{c}}
\newcommand{\RPs}[1]{\RP_{\!#1}}
\newcommand{\RPsc}[1]{\RP^{\mathsf{c}}_{\!#1}}
\newcommand\IP{\M{IP}}
\newcommand\MP{\M{MP}}
\newcommand\AC{\M{AC}}
\newcommand\AU{\M{AU}}
\newcommand\AX{\M{AX}}
\newcommand\SOSO{S1S$_0$}
\newcommand\modelsUP{\models_{\M{UP}}}

\newcommand\lts{\lhd}
\newcommand\ex[2]{\exists_{#1}#2}
\newcommand\trans[3]{#1\Rightarrow^{#2}_A#3}
\newcommand\transa[3]{#1\Rrightarrow^{#2}_A#3}
\newcommand\col{\M{C}}
\newcommand\equivUP{\equiv_{\M{UP}}}

\newcommand{\substr}[3]{#1_{#2}^{#3}}

\newcommand{\sing}[1]{\M{sing}~#1}

\newcommand{\sgsum}[0]{\col}
\newcommand{\merge}[3]{#2 \sim_{#1} #3}
\newcommand{\notmerge}[3]{#2 \not\sim_{#1} #3}
\newcommand{\fullsat}[2]{#1 \models #2}

\numberwithin{lemma}{section}

\makeatletter
\newcommand{\oset}[3][0ex]{%
  \mathrel{\mathop{#3}\limits^{
    \vbox to#1{\kern-2\ex@
    \hbox{$\scriptstyle#2$}\vss}}}}
\makeatother
\newcommand{\transition}[3]{#1 \oset[-0.1ex]{#2}{\to} #3}
\newcommand{\pred}[1]{\ensuremath{p_{#1}}}

\newcommand\dropi[1]{#1_{i..}}

\begin{document}
\title{Constructive Analysis of\\ S1S and Büchi Automata}
\author{Moritz Lichter and Gert Smolka\\
  Saarland University\\
  \normalsize\today}
\date{}
\maketitle

\begin{abstract}
  \noindent
  We study S1S and Büchi automata 
  in the constructive type theory 
  of the Coq proof assistant.  
  For UP semantics (ultimately periodic sequences), 
  we verify Büchi's translation of formulas to automata 
  and thereby establish decidability of S1S constructively.  
  For AS semantics (all sequences), 
  we verify Büchi's translation assuming that 
  sequences over finite semigroups have Ramseyan factorisations (RF).  
  Assuming RF, UP semantics and AS semantics agree.  
  RF is a consequence of Ramsey's theorem 
  and implies the infinite pigeonhole principle, 
  which is known to be unprovable constructively.  
  We show that each of the following properties holds for UP semantics 
  but is equivalent to RF for AS semantics: 
  excluded middle of formula satisfaction, 
  excluded middle of automaton acceptance, and 
  existence of complement automata.
\end{abstract}

\section{Introduction}

S1S is the monadic second-order logic of order
with first-order variables 
ranging over natural numbers and
second-order variables 
ranging over possibly infinite sets of numbers.
Following Büchi~\cite{buechi1962},
decidability of S1S can be shown with 
a compositional translation of formulas to automata
realizing constructions of formulas
with operations on automata.
The automata employed by the translation
are NFAs accepting infinite sequences.
One speaks of Büchi automata
to indicate a particular acceptance
condition for infinite sequences formulated by Büchi.
The reduction of formulas to automata
works well for S1S and various other logics,
including temporal 
logics~\cite{DemriGoronkoLange,Perrin04,thomas97}.

We study S1S and 
Büchi's translation to automata
in the constructive type theory 
of the Coq proof assistant~\cite{coq}.
Coq's type theory extends Martin-Löf type theory 
with an impredicative universe of propositions
such that excluded middle
can be assumed consistently
for all propositions.
This matters for our purposes
since several aspects of S1S
and Büchi's translation
cannot be verified constructively.

We represent sequences over
a type $A$ as functions $\nat\to A$ from natural numbers to $A$
and sets of numbers as boolean sequences
(functions $\nat\to\bool$).
An automaton accepts a sequence
if there is a run on the sequence 
that passes through accepting states infinitely often.

When we verify the operations on automata
needed for the translation of formulas, 
all operations but Büchi's complement operation
can be verified constructively.
The verification of Büchi's complement operation
requires a restricted form of Ramsey's theorem
known as additive Ramsey theorem~\cite{Kolodziejczyk2016,riba12}.
We refer to the property asserted by the theorem as RA.  
If we assume RA,
we can verify the translation of formulas into automata
and thereby show that S1S is decidable.

We will mostly work with a property RF 
that is constructively equivalent to RA.
The verification of Büchi's complement operation
is constructive except a single spot
where an instance of RF is needed.
RF says that every sequence 
over a finite semigroup of colors
has a factorisation $u_0,u_1,u_2,\dots$
such that all strings $u_1,u_2,u_3\dots$
have the same color.
Following Blumensath~\cite{AB15} we call factorisations with this property 
Ramseyan factorisations.
Variants of Ramseyan factorisations appear in the
literature~\cite{Perrin04,Shelah75}.

We show that RF implies the infinite pigeonhole principle
(for every sequence over a finite type there is an element occurring infinitely often),
which is known to be unprovable constructively~\cite{VeldmanBezem92}.
It follows that RF and RA
are unprovable constructively, too.

Let FX be the property that 
the satisfaction relation $I\models\phi$
between interpretations and formulas of S1S
satisfies XM (excluded middle):
$\forall I\phi.~I\models\phi\lor I\nvDash\phi$.
Note that FX is a special instance
of general excluded middle.
A~main result of this paper is a proof 
that FX and RA are equivalent constructively.
As a consequence, 
we know that RA is necessary and sufficient
to correctly formalise S1S in constructive type theory.

We provide two further characterisations of RF
by showing that RF is constructively equivalent
to AC (complement automata exist) and
AX (acceptance by automata satisfies XM).
AC is interesting since it implies that no
complement operation can be verified constructively.

We refer to the standard semantics 
of automata and formulas introduced so far
as \textit{AS semantics} (for all sequences) 
to distinguish it from 
an alternative semantics we call 
\textit{UP semantics}.
UP semantics~\cite{BresolinMP09,calbrix1993,calbrix1994}
is based on ultimately periodic sequences $xy^\omega$,
finitely specified with two strings $x$ and $y$.
Since UP sequences over finite semigroups
obviously have  Ramseyan factorisations,
correctness of Büchi's complement operation 
for UP sequences can be verified constructively.

We show that in constructive type theory
S1S with UP semantics is decidable and classical 
(i.e., UP satisfaction of formulas satisfies XM).
This shows that UP semantics provides a
purely constructive formalisation of S1S.
This is in contrast to AS semantics,
which requires RF to adequately formalise S1S.
We show that, given RF,
UP semantics agrees with AS semantics
as it comes to satisfiability of formulas.

We provide one further 
constructive characterisation of RF
we call AU.  AU says that two automata 
accept the same sequences if
they accept the same UP sequences.
AU is known to hold 
classically~\cite{calbrix1993,calbrix1994}.

There is a remarkable coincidence between
our work and the work of 
Ko{\l}odziejczyk et al.~\cite{Kolodziejczyk2016}
who study the translation of S1S formulas
to automata in RCA$_0$, a system 
of weak second-order arithmetic
(a classical logic satisfying XM).
They show that the following properties
are pairwise equivalent in RCA$_0$:
correctness of Büchi complementation,
decidability of S1S, and the additive Ramsey theorem.

We spend considerable effort on proving 
that both FX and AX imply RF.
For this we establish 
a further constructive characterisation of RF
we call RP for Ramseyan pigeonhole principle.
RP has a straightforward classical proof and
can be related to satisfaction of S1S formulas
and Büchi acceptance of automata.
RP is based on a relation for sequences
over finite semigroups appearing as merging relation in the 
literature~\cite{buechi1973,mcNaughton66,Perrin04,riba12,Shelah75}.
The relation is used in the literature
to prove Ramseyan properties similar to RA and RF
using excluded middle.

\paragraph{Organisation of the paper}
The paper is written at a level of abstraction 
that does not require detailed knowledge of
constructive type theory.
We start with preliminaries concerning
type theory, sequences, and Ramseyan factorisations
and show that RF and RA are equivalent.
We postpone the definition of full S1S
and start with minimal S1S providing  
the basis for the translation to automata.
We review Büchi automata
and show that all operations but complement
can be verified constructively.
In particular,
we verify the correctness
of the operation for existential quantification
for UP semantics.
We then show correctness of Büchi's complement operation,
both for UP sequences (no assumption needed)
and all sequences (RF needed).
We now show that RF, AC, and AU 
are constructively equivalent
and obtain the decidability results
for minimal S1S for both
AS semantics and UP semantics.

We then define full S1S and reduce it to minimal S1S.
What is missing at this point are proofs that
FX and AX imply RF.
For this purpose we introduce RP and show
that it is  equivalent to RF.
We then show that both FX and AX both entail RP.

\paragraph{Coq development}
There is a Coq development 
proving all results of the paper.
Instead of defining AS and UP semantics of S1S separately,
we work with a generalisation,
which can be instantiated for AS and UP semantics.
The Coq development is available at
\url{http://www.ps.uni-saarland.de/extras/S1S}.
The definitions and statements in this
paper are hyperlinked with our Coq development
available for browsing on our project
web page.
\section{Preliminaries}

In constructive type theory, 
the law of excluded middle (XM)
is not built-in and thus propositions 
like $P\lor\neg P$ and $\neg\neg P\to P$
are not trivially provable.
Moreover, there are no native sets,
and functions must be total and
can only be defined computationally.

We write \emph{$\one$} and \emph{$\bool$} 
for the inductive types providing 
the single value~\emph{$\oneM$}
and the boolean values \emph{\M{true}} and \emph{\M{false}},
respectively.
We write \emph{$\nat$} for the inductive type
providing the numbers ${0,\,S0,\,S(S0),\,\ldots\,\,}$.
Note that $Sn=n+1$.  
The letters~$i$, $j$, $k$, $l$, $m$, and $n$
will range over numbers.

We write \emph{$\exists^\omega n.~pn$}
for $\forall k\,\exists n.\,n\ge k\land pn$
and say that
\emph{$p$ holds infinitely often}.
Moreover, 
we write \emph{$\exists n\ge k.\,pn$}
for $\exists n.\,n\ge k\land pn$,
and \emph{$\forall n\ge k.\,pn$}
for $\forall n.\,n\ge k\to pn$.

Let $p$ be a unary predicate on a type $A$.
We say that \emph{$p$ satisfies XM} 
if ${\forall a.~pa\lor\neg pa}$,
and that \emph{$p$ is decidable}
if we have a function $f:A\to\bool$
such that ${\forall a.~pa\toot fa=\btrue}$.
For propositions and predicates with $n\ge2$ arguments
satisfaction of XM and decidability
are defined analogously.
Note that decidable propositions and predicates satisfy XM.  
Since functions definable in constructive type theory
are computable,
decidable predicates are computationally decidable.

We will make use of the fact that
the predicates $\lam{ij}{i<j}$ and $\lam{ij}{i\le j}$ 
are decidable (${i,j:\nat}$).


We will occasionally use the proposition
\begin{align*}
  \N{\XM}&~:=~\forall P.~P\lor\neg P
\end{align*}
which states
that every proposition satisfies XM.
Note that every predicate satisfies XM
if we assume \XM.
Assuming \XM\ in Coq is consistent and
does not change our notion of decidability,
as functions on non propositional types 
stay computable.
Given a proposition $P$, 
we write $\N{\xm(P)}:=P\lor\neg P$.

\begin{fact}
  \label{fact-xm-rules}~
  For all propositions $P$, $Q$ and all predicates $p$:
  \begin{enumerate}
  \item $\neg(P\land Q)\toot(P\to \neg Q)$.
  \item $\neg(\exists x.px)\toot\forall x.\neg px$. 
  \item $\xm(P)\to(P\toot\neg\neg P)$. 
  \item $\xm(\exists x.px)\to((\exists x.px)\toot\neg\forall x.\neg px)$.
  \end{enumerate}
\end{fact}

A \emph{sequence} over a type $A$ 
is a function $\sigma:\nat\to A$.
If $\sigma n=a$, 
we say that \emph{$\sigma$ is~$a$ at position $n$}.
We will use the notation $\N{A^\omega}:=\nat\to A$
for the type of sequences over $A$.

Two sequences $\sigma$ and $\tau$ over $A$ \emph{agree}
if $\sigma n=\tau n$ for all $n$.
We write \emph{$\sigma\equiv\tau$} 
to say that $\sigma$ and $\tau$ agree.
We also say that 
two sequences are \emph{equivalent}
if they agree.
The notion of agreement is needed since we work
in a non-extensional type theory.

A \emph{boolean sequence} is a sequence over $\bool$.
The letter $\beta$ will range over boolean sequences.
A boolean sequence may be seen as a decidable set of numbers.
Following this interpretation, 
we write \emph{$n\in\beta$} for $\beta n=\btrue$.
A boolean sequence is called \emph{infinite} if
it is infinitely often $\btrue$ 
(i.e., $\exists^\omega n.\,n\in\beta$),
and \emph{nonempty} if $\exists n.\,n\in\beta$.

In Coq's constructive type theory one can
(computationally) obtain
an element for a nonempty sequence.
This fact is known as \emph{constructive choice}.
We will repeatedly make use of constructive choice
when we construct functions.
Here is a more precise formulation of this fact.

\begin{fact}[Constructive Choice]\label{fact-cc}
  One can define a function that given 
  a boolean sequence~$\beta$
  and a proof of\, $\exists n.~n\in\beta$ 
  yields an $n\in\beta$.
\end{fact}

\noindent A function $f:\nat\to\nat$ \emph{enumerates}
a boolean sequence $\beta$
if ${\forall n.~n\in\beta\toot\exists k.\,fk=n}$.

\begin{fact}\label{fact-bseq-enumeration}
  For every infinite boolean sequence $\beta$
  one can obtain a strictly monotone function
  enumerating $\beta$.
  Moreover, 
  for every strictly monotone function $f:\nat\to\nat$
  one can obtain an infinite boolean sequence $\beta$
  such that $f$ enumerates $\beta$.
\end{fact}
\begin{proof}
  The first claim follows with constructive choice.
  The second claim follows since 
  $\exists k.\,fk=n$ is equivalent to 
  $\exists k\le n.\,fk=n$ for strictly monotone $f$.
\end{proof}

\emph{Strings} over a type $A$
are provided with an inductive type \emph{$A^+$} 
defined with two constructors:
\begin{align*}
  x:A^+&~::=~a\mid ax\qquad (a:A)
\end{align*}
Our definition does not provide an empty string,
since this is advantageous for the purposes of this paper.
When we say \textit{string over~$A$}
we will always mean an element of $A^+$.
We write $xy$ for the \emph{concatenation}
of two strings $x$ and $y$.

Given two strings  $x$ and $y$ over $A$,
we write \emph{$xy^\omega$} for a sequence
agreeing with the infinite concatenation $xyyy\cdots$.
Formally, we define the sequence $xy^\omega:A^\omega$ 
by recursion on numbers:
\begin{align*}
  ay^\omega(0)&~=~a
  &ay^\omega(Sn)&~=~yy^\omega(n)\\
  (ax)y^\omega(0)&~=~a
  &(ax)y^\omega(Sn)&~=~xy^\omega(n)
\end{align*}
We call a pair $(x,y)$ of two strings over~$A$ 
a \emph{UP sequence} over $A$.
Notationally, we will identify 
the pair $(x,y)$ with the sequence $xy^\omega$.

A \emph{discrete type} is a type $X$
together with a boolean function
deciding equality on~$X$.
%
A \emph{finite type} is a discrete type $X$
together with a duplicate-free list containing
all elements of $X$.  
The letters \emph{$\Sigma$} 
and~\emph{$Q$} will range
over finite types.
The type $\one$ and $\bool$ 
can be accommodated as finite types
and $\nat$ can be accommodated as discrete type.
Moreover, finite types are closed under taking
\emph{products $\Sigma_1\times\Sigma_2$}
and \emph{sums $\Sigma_1+\Sigma_2$}.
Given a finite type $\Sigma$, 
there is a finite type \emph{$2^\Sigma$}
containing exactly the sets over $\Sigma$.
The sets in $2^\Sigma$ have decidable membership.

\section{Ramseyan Factorisations}
\setCoqFilename{Buechi.FiniteSemigroups}

A \emph{factorisation} of a sequence $\sigma$ over $A$
is a sequence $\tau$ over $A^+$
such that $\sigma$ agrees with 
the infinite concatenation 
$(\tau0)(\tau1)(\tau2)\cdots$.

A \coqlink[FiniteSemigroup]{\emph{finite semigroup}} 
is a finite type $\Gamma$ 
together with an associative operation~$+$.
We will call the elements of finite semigroups \emph{colors}
and define the \emph{color of a string} over $\Gamma$ as follows:
\begin{align*}
  \col(a_0\cdots a_n)&~:=~a_0+\cdots+a_n
\end{align*}
We have $\col(xy)=\col x+\col y$ 
for all strings $x$ and $y$ over $\Gamma$.
Note that string concatenation 
is a semigroup operation for $\Gamma^+$
and that $\col$ is a semigroup morphism
$\Gamma^+\to\Gamma$.
The letter $\Gamma$ will range over finite semigroups.

\setCoqFilename{Buechi.RamseyanFactorizations}
A \coqlink[is_ramseyan_factorization]{\emph{Ramseyan factorisation}} of 
a sequence $\sigma$ over $\Gamma$ 
is a factorisation~$\tau$ of $\sigma$
such that all strings $\tau1,\tau2,\tau3,\ldots$
have the same color.
We call $\col(\tau1)$ the \emph{color of the factorisation}.
Note that $\tau0$ and $\tau1$ may have different colors.
We define the proposition \RF\ as follows:
\begin{itemize}
\item[\emph{\coqlink[RF]{\RF}} $:=$] 
  Every sequence over a finite semigroup 
  has a Ramseyan factorisation.
\end{itemize}

\begin{fact}[][up_admits_ramseyan_fac]
  \label{fact-UP-RF}
  Every UP sequence over a finite semigroup 
  has a Ramseyan factorisation.
\end{fact}
\begin{proof}
  The sequence $x,y,y,y,\ldots\,$ 
  is a Ramseyan factorisation of~$xy^\omega$.
\end{proof}

We will now show that \RF\ implies the infinite pigeonhole principle 
and Markov's principle.
The infinite pigeonhole principle is unprovable constructively~\cite{VeldmanBezem92}
and Markov's priniciple is unprovable in CIC~\cite{coquand16,pedrot18},
a type theory similar to the one of Coq. Hence \RF\ is unprovable, too.

We define two propositions expressing 
the \textit{infinite pigeonhole principle}
and \textit{Markov's principle}:
\begin{itemize}
\item[\emph{\coqlink[IP]{\IP}} $:=$] 
  For every finite type $\Sigma$
  and every sequence $\sigma$ over $\Sigma$
  there exists a value 
  $a:\Sigma$ such that $\exists^\omega n.~\sigma n=a$.
\item[\emph{\coqlink[MP]{\MP}} $:=$] 
  If a boolean sequence is not constantly false,
  then there exists a position where it is true.
\end{itemize}

\begin{fact}\label{fact-IP-MP}
  \coqlink[RF_implies_IP]{\RF\ implies \IP\ } and \coqlink[IP_implies_MP]{\IP\ implies \MP}.
\end{fact}
\begin{proof}
  $\RF\to\IP$.
  Assume \RF\ and 
  let $\sigma$ be a sequence over a finite type $\Sigma$.
  We fix $\lam{ab}a$ as semigroup operation on $\Sigma$.
  By \RF\ we have a Ramseyan factorisation $\tau$ of~$\sigma$.
  Thus there is a color $a:\Sigma$
  such that the first symbol of $\tau n$ is $a$
  for all $n\ge1$.  
  Hence $\exists^\omega n.~\sigma n=a$.

  $\IP\to\MP$.
  Assume \IP\ and
  let $\beta$ be a boolean sequence such that 
  ${\neg\forall n.~\beta n=\bfalse}$.
  Let $\beta'$ be the boolean sequence such that
  $\beta' n$ is the boolean disjunction of $\beta0,\ldots,\beta n$.
  By \IP\ we have a boolean value $b$ 
  such that $\exists^\omega n.~\beta' n=b$.
  If $b=\btrue$, we have a position where~$\beta$ is $\btrue$.
  If $b=\bfalse$, we can show that $\beta$ is constantly $\bfalse$,
  which contradicts the assumption.
\end{proof}

\setCoqFilename{Buechi.AdditiveRamsey}

Next we show that \RF\ is equivalent to
a proposition \RA\ 
expressing a weakening of Ramsey's theorem 
called \textit{additive Ramsey theorem} 
in Ko{\l}odziejczyk et al.~\cite{Kolodziejczyk2016}
and \textit{Ramsey's theorem with additive coloring}
in Riba~\cite{riba12}.
The definition of \RA\ requires some preparation.

Let $\Gamma$ be a finite semigroup
and $\gamma$ be a function $\nat\to\nat\to\Gamma$.
We call $\gamma$ an \coqlink[additive]{\emph{additive coloring into $\Gamma$}}
if $\gamma ij+\gamma jk=\gamma ik$ for all $i<j<k$.
A boolean sequence~$\beta$ 
is called \coqlink[homogenous]{\emph{homogeneous for $\gamma$}}
if there exists a color $c$ such that
$\gamma ij=c$ for all $i,j\in\beta$ such that $i<j$.
We now define the proposition \RA\ as follows:
\begin{itemize}
\item[\emph{\coqlink[RA]{\RA}} $:=$] 
  For every additive coloring into a finite semigroup
  there exists an infinite and homogeneous boolean sequence.
\end{itemize}

Given a sequence $\sigma$ and numbers $i<j$,
we write \emph{$\substr\sigma ij$}
for the \emph{substring of~$\sigma$} that
starts at position $i$ and ends at position $j-1$
(i.e., position $i$ is inclusive and 
position $j$ is exclusive).
We realise the notation with
a polymorphic function 
${\forall A.~A^\omega\to\nat\to\nat\to A^+}$
that yields $\substr\sigma ij$ for $i<j$.

Let $\Gamma$ be a finite semigroup.
A sequence $\sigma$ over $\Gamma$ 
may be represented as the additive coloring
$\lam{ij}{\,\col(\substr\sigma ij)}$,
and the relevant part (i.e., $i<j$)
of an additive coloring~$\gamma$ into $\Gamma$
may be represented as the sequence 
$\lam n{\gamma n(Sn)}$.

A factorisation $\tau$ of a sequence $\sigma$
may be represented as the infinite set 
of the starting positions
of the factors $\tau1,\tau2,\ldots$ in $\sigma$.
We represent this set as an infinite boolean sequence.
By Fact~\ref{fact-bseq-enumeration}
we can obtain for a factorisation 
a corresponding infinite boolean sequence,
and for an infinite boolean sequence 
a corresponding factorisation.

\setCoqFilename{Buechi.FiniteSemigroups}
An element $a$ of a semigroup is \coqlink[idempotent]{\emph{idempotent}} if $a+a=a$.
It is well-known~\cite{Perrin04}
that for every element $a$
of a finite semigroup there
exists a number $n$ such that
$n\cdot a$ is idempotent
($n\cdot a$ is notation 
for the sum $a+\dots+a$
with $n$ summands).

\setCoqFilename{Buechi.RamseyanFactorizations}
\begin{fact}[][admits_ramseyan_fac_iff_idem_ramseyan_fac]
  \label{fact-RF-idempotence}
  Let $\sigma$ be a sequence over a finite semigroup $\Gamma$
  that has a Ramseyan factorisation.
  Then $\sigma$ has a Ramseyan factorisation 
  with an idempotent color.
\end{fact}
\begin{proof}
  Let $\tau$ be a Ramseyan factorisation of $\sigma$.
  Since $\Gamma$ is finite,
  there exists some number $n$ such that
  $n\cdot\col(\tau1)$ is idempotent.
  Since all factors $\tau1,\tau2,\tau3,\ldots$ 
  have identical color,
  we obtain an idempotent Ramseyan factorisation of~$\sigma$
  by successively merging $n$ adjacent factors of $\tau$
  into a single factor.
\end{proof}

\setCoqFilename{Buechi.AdditiveRamsey}

\begin{fact}[][ramsey_fac_iff_homogenous]
  \label{fact-RF-inf-homo}
  Let $\sigma$ be a sequence over a finite semigroup $\Gamma$.
  Then $\sigma$ has a Ramseyan factorisation
  if and only if there exists
  an infinite boolean sequence 
  that is homogeneous for 
  $\lam{ij}{\,\col(\substr\sigma ij)}$.
\end{fact}
\begin{proof}
  Let $\sigma$ have an Ramseyan factorisation.
  By Fact~\ref{fact-RF-idempotence} we have
  a Ramseyan factorisation $\tau$ for $\sigma$
  that has an idempotent color~$c$.
  Now the infinite boolean sequence 
  representing $\tau$ satisfies the claim.
  
  Let~$\beta$ be an infinite boolean sequence
  that is homogeneous for
  $\lam{ij}{\,\col(\substr\sigma ij)}$.
  Then the factorisation represented by $\beta$ 
  is a Ramseyan factorisation of $\sigma$.
\end{proof}

\begin{fact}[][RF_implies_RA]
  \label{fact-ra-implies-rf}
  \RA\ implies \RF.
\end{fact}
\begin{proof}
  Assume \RA\ and let 
  $\sigma$ be a sequence over a finite semigroup~$\Gamma$.
  By Fact~\ref{fact-RF-inf-homo} it suffices to show
  that there is  an infinite boolean sequence that
  is homogeneous for $\lam{ij}{\,\col(\substr\sigma ij)}$.
  This follows with \RA\ since
  $\lam{ij}{\,\col(\substr\sigma ij)}$
  is an additive coloring into $\Gamma$.
\end{proof}

\begin{fact}[][RA_implies_RF]
  \label{fact-rf-implies-ra}
  \RF\ implies \RA.
\end{fact}
\begin{proof}
  Assume \RF\ and 
  let $\gamma$ be an additive coloring
  into a finite semigroup $\Gamma$.
  We show that there exists
  an infinite boolean sequence~$\beta$
  that is homogeneous for $\gamma$.
  We consider the sequence $\sigma n:=\gamma n(Sn)$.
  By \RF\ and Fact~\ref{fact-RF-inf-homo} 
  there exist an infinite $\beta$
  that is homogeneous for $\lambda i j. \col(\substr{\sigma}{i}{j})$.
  The claim follows since 
  $\gamma ij =\col(\substr\sigma ij)$ for all $i<j$.
\end{proof}

\setCoqFilename{Buechi.MinimalS1S}
\section{Minimal S1S}

We consider a minimal variant of S1S that has
no first-order variables.  
Full S1S with both kinds of variables reduces
to the minimal variant of S1S we consider.
We shall use the shorthand \emph{\SOSO}{} for minimal S1S.
Other variants of S1S 
not using first-order variables 
appear in~\cite{AB15,thomas97}.

We start with a finite type \emph{$\var$} of \emph{variables}
and formalise the syntax of minimal S1S
with an inductive type of \emph{formulas}:
\begin{equation*}
  \coqlink[MinForm]{\varphi,\psi}~::=~
  X\lts Y\mid 
  X\incl Y\mid
  \varphi\land\psi\mid
  \neg\varphi\mid
  \exists X.\varphi
  \qquad
  (X,Y:\var)
\end{equation*}
Informally speaking, variables range over sets represented
as boolean sequences.  A formula $X\lts Y$ says that there
are numbers $m<n$ such that $m\in X$ and ${n\in Y}$.
Moreover, a formula $X\incl Y$ says that~$X$ is a subset of $Y$.

We now formally define the \emph{AS~semantics} of \SOSO.
An \emph{interpretation} is a sequence over $2^\var$
(the finite type containing all sets of variables).
We write $\sigma_X$ for the boolean sequence such that 
${\forall n.~\sigma_Xn=\btrue\toot X\in\sigma n}$.
Note that $\sigma_X$ represents the set for the variable~$X$.
We define interpretations as sequences 
so that we can translate a formula into an automaton
that accepts exactly the sequences satisfying the formula.
The letter $\sigma$ will range over interpretations in the following.

We define the satisfaction relation \coqlink[satisfies]{\emph{$\sigma\models\phi$}}
between interpretations and formulas by recursion on formulas
such that the following equivalences trivially hold:
\begin{align*}
  \sigma\models X{\lts}Y
  &~\toot~\exists mn.~m<n\land m\in\sigma_X\land n\in \sigma_Y
  \\
  \sigma\models X{\incl}Y
  &~\toot~\forall n.~n\in\sigma_X\to n\in\sigma_Y
  \\
  \sigma\models \phi{\land}\psi
  &~\toot~\sigma\models \phi\land \sigma\models\psi
  \\
  \sigma\models \neg\phi
  &~\toot~\neg(\sigma\models\phi)
  \\
  \sigma\models \exists X.\phi
  &~\toot~\exists\tau.~\tau\models\phi\land\sigma\approx_X\tau
\end{align*}
The notation \coqlink[equiv_apart]{\emph{$\sigma\approx_X\tau$}} stands for
$\forall Z.~Z=X\lor \sigma_Z\equiv\tau_Z$
and says that $\sigma$ and $\tau$ agree 
for all variables but possibly $X$.


We define the \emph{UP semantics} for \SOSO.
Everything stays as it is 
except that all sequences over $2^\var$
are replaced with UP sequences over~$2^\var$.
This is also the case for 
the existentially quantified sequence $\tau$.
Recall that a UP sequence is a \textit{pair} 
of two strings $x$ and $y$ that is
interpreted as the sequence $xy^\omega$.
We will write \emph{$xy^\omega\modelsUP\phi$}
for the satisfaction relation of \SOSO\
with UP semantics.

\section{Translation to Abstract Automata}

Many aspects of the translation of formulas 
to automata can be explained
without knowing the details of automata.
We will therefore work
with an abstract type of \coqlink[AbstractAutomata]{\emph{automata}}
in this section.
The letters~$A$ and $B$ will range over 
automata of this type,
and the letters $\sigma$ and~$\tau$ will range over 
sequences over $2^\var$.
We assume an 
\emph{acceptance relation \coqlink[lang]{$\sigma\models A$}}
between sequences and automata.
We read $\sigma\models A$ 
as~\emph{$\sigma$~satisfies $A$}
or as \emph{$A$~accepts $\sigma$}.
We also assume functions
\emph{$\M{A}^\incl_{XY}$},  \
\emph{$\M{A}^\lts_{XY}$}, \
\emph{$A\cap B$}, \
\emph{$\ol A$}, and 
\emph{$\ex X A$}
on automata mimicking the constructors for formulas.
We refer to these functions as \textit{operations}
and assume they have the following properties.
\begin{enumerate}
\coqitem[less_aut_correct] 
  $\sigma\models \M{A}^\lts_{XY}~\toot~
  \exists mn.~m<n\land
  m\in\sigma_X\land 
  n\in \sigma_Y$
\coqitem[incl_aut_correct] 
  $\sigma\models \M{A}^\incl_{XY}~\toot~
  \forall n.~n\in\sigma_X\to n\in\sigma_Y$
\coqitem[intersect_aut_correct] 
  $\sigma\models A\cap B~\toot~\sigma\models A\land \sigma\models B$
\coqitem[complement_aut_disjoint]
  $\sigma\models A\to\sigma\models\ol A\to\bot$
\coqitem[complement_aut_exhaustive] $\sigma\models A\lor\sigma\models\ol A$
\coqitem[ex_aut_correct] 
  $\sigma\models\ex XA~\toot~\exists\tau.~
  \tau\models A\land\sigma\approx_X\tau$
\end{enumerate}
Note that the specification of complement automata
deviates from the specification of the other operations.
It consists of two assumptions (4) and (5)
we call \emph{disjointness} and \emph{exhaustiveness}.
The specification of complement automata
with disjointness and exhaustiveness 
rather than a single equivalence 
is necessitated by the constructive analysis.

\begin{fact}\label{fact-abstract-automata-disj_exh}
  \coqlink[complement_aut_correct]{$\sigma\models\ol A\toot\sigma\nvDash A$}
  and
  \coqlink[aut_lang_sat_xm]{$\sigma\models A\lor\sigma\nvDash A$}.
\end{fact}

Together, the two statements of 
Fact~\ref{fact-abstract-automata-disj_exh}
are equivalent to disjointness and exhaustiveness
of complement automata. However, the first statement
of Fact~\ref{fact-abstract-automata-disj_exh}
does not suffice for exhaustiveness constructively.

We now define a function \coqlink[formula_aut]{\emph{$\alpha$}}
translating formulas into automata:
\begin{align*}
  \alpha(X\lts Y)
  &~=~\M{A}^\lts_{XY}
  &\alpha(\phi\land\psi)
  &~=~\alpha(\phi)\cap\alpha(\psi)
  \\
  \alpha(X\incl Y)
  &~=~\M{A}^\incl_{XY}
  &\alpha(\neg\phi)
  &~=~\ol{\alpha(\phi)}
  \\
  &&\alpha(\exists X.\phi)
  &~=~\ex {X}{(\alpha(\phi))}
\end{align*}

\begin{fact}[][formula_aut_correct]
  \label{fact-absaut-trans}
  Let automata and operations
  satisfying the assumptions stated above be given.
  Then ${\sigma\models\phi\toot\sigma\models\alpha(\phi)}$.
  Thus formala satisfaction satisfies XM
  and satisfiability of formulas is decidable
  if satisfiability of automata is decidable.
\end{fact}

\section{Büchi Automata}

We now consider NFAs with Büchi acceptance.
It turns out that all operations but complement
can be defined and verified constructively
following familiar ideas.
Care must be taken with the formulation
of Büchi acceptance.
We require 
that accepting states are visited infinitely often.
Constructively, this is weaker than requiring
that a single accepting state is visited infinitely often,
as it is sometimes done in the literature~\cite{thomas97}.
Choosing the weak version 
is of particular importance for the complement operation,
which we will consider in a later section.

The usual method for 
deciding satisfiability of Büchi acceptance
(non-emptiness problem)
can be verified constructively.
In fact, the method yields a satisfying UP sequence
in case the NFA is satisfiable.

As it comes to UP sequences and UP semantics,
only the operation for existential quantification
needs special attention.
We introduce the notion of a match 
to deal with this issue.

\setCoqFilename{Buechi.NFAs}
We formalise 
\coqlink[NFA]{\emph{NFAs}} (nondeterministic finite automata) 
over a finite type $\Sigma$ called \emph{input alphabet}
as tuples consisting of
a finite type of \emph{states},
a decidable \emph{transition relation},
and decidable predicates 
identifying \emph{initial} and \emph{accepting} states.

\setCoqFilename{Buechi.Buechi}
Let $A$ be an NFA over $\Sigma$ with state type $Q$.
A sequence $\rho$ over $Q$ \emph{admits} 
a sequence~$\sigma$ over $\Sigma$
if $(\rho n,\sigma n,\rho(Sn))$ 
is a transition of $A$ for all $n$.
A \emph{run on $\sigma$}
is a sequence $\rho$ over $Q$
that starts with an initial state and admits $\sigma$.
A run is \coqlink[accepting]{\emph{accepting}} if 
it passes infinitely often through accepting states.

An NFA $A$ over $\Sigma$
\coqlink[accepts]{\emph{accepts}} a sequence~$\sigma$ over $\Sigma$
if $A$ has an accepting run on $\sigma$.
We write \emph{$\sigma\models A$} 
if $A$ accepts $\sigma$.
We also say that $\sigma$ satisfies $A$
if $A$ accepts $\sigma$, 
and  that $A$ is \emph{satisfiable} 
if $A$ accepts some sequence.

Recall that the type theory we are working in 
does not provide functional extensionality.
This does not hurt since automata access
sequences only pointwise.



\setCoqFilename{Buechi.MinimalS1S}

\begin{fact}[][incl_nfa_correct]
  \label{fact-NFA-incl}
  One can define a function
  that given two variables $X$ and~$Y$ yields 
  an NFA~\emph{$\M{A}^\incl_{XY}$} over $2^\var$ 
  such that
  $\sigma\models \M{A}^\incl_{XY}$ if and only if
  $\forall n.~n\in\sigma_X\to n\in\sigma_Y$.
\end{fact}

\begin{fact}[][less_nfa_correct]
  \label{fact-NFA-lt}
  One can define a function
  that given two variables $X$ and~$Y$ yields 
  an NFA~\emph{$\M{A}^\lts_{XY}$} over $2^\var$ 
  such that
  $\sigma\models \M{A}^\lts_{XY}$ if and only if
  $\exists mn.~m<n\land
  m\in\sigma_X\land 
  n\in \sigma_Y$.
\end{fact}

\setCoqFilename{Buechi.Buechi}

\begin{fact}[][exact_up_nfa_correct]
  \label{fact-NFA-UP}
  One can define a function
  that given two strings $x$ and $y$ over $\Sigma$
  yields an NFA \emph{$\M{A}_{xy^\omega}$} over $\Sigma$
  accepting exactly the sequences equivalent to $xy^\omega$.
\end{fact}

\setCoqFilename{Buechi.NFAs}

Let $A$ be an NFA over $\Sigma$ with state type $Q$.
Given a string $u$ over~$Q$ and a string~$x$ over $\Sigma$,
we say that $u$ is a \emph{path on $x$} if
$x$ provides symbols yielding transitions
between adjacent states of $u$.
We say that \emph{$A$ accepts $x$} 
if $A$ has a path on $x$ 
starting with an initial state
and ending with an accepting state.
We define two decidable predicates:
\begin{align*}
\coqlink[transforms]{\trans pxq} :=&
  \text{ } A \text{ has a path on } x \text{ from } p \text{ to } q \text{.}\\
\coqlink[transforms_accepting]{\transa pxq} :=&
  \text{ } A \text{ has a path on } x \text{ from } p \text{ to } q  
  \text{ passing through}\\& \text{ an accepting state 
  before the last position.}
\end{align*}
\setCoqFilename{Buechi.Buechi}
A pair $(x,y)$ is a \coqlink[is_match]{\emph{match}} of $A$ if
there exist states $p$ and $q$ such that
$p$ is initial, $\trans pxq$, and $\transa qyq$.

\begin{fact}
  \label{fact-NFA-solver}
  ~
  \begin{enumerate}
  \coqitem[match_accepted] If $(x,y)$ is a match of an NFA $A$, then $xy^\omega$ satisfies $A$.
  \coqitem[nonempty_iff_match] An NFA is satisfiable if and only if it has a match.
  \coqitem[dec_exists_match] It is decidable whether an NFA has a match.
  \coqitem[obtain_match] If an NFA has a match, one can obtain a match.
  \coqitem[dec_buechi_empty_informative] It is decidable whether an NFA is satisfiable.
  \item In case an NFA $A$ is satisfiable, one can obtain a match of $A$.
  \end{enumerate}
\end{fact}

\begin{fact}[][union_correct]
  \label{fact-NFA-union}
  One can define a function
  that given two NFAs $A$ and $B$ over~$\Sigma$ yields 
  an NFA \emph{$A\cup B$} over $\Sigma$ with
  $\sigma\in A\cup B\toot\sigma\models A\lor\sigma\models B$.
\end{fact}
\begin{proof}
  Define $A\cup B$ as the disjoint union of $A$ and $B$.
\end{proof}

\begin{fact}\label{fact-NFA-inter}
  One can define a function
  that given two NFAs $A$ and $B$ over~$\Sigma$ yields 
  an NFA \emph{$A\cap B$} over $\Sigma$ such that:
  \begin{enumerate}
  \coqitem[intersect_correct] $\sigma\in A\cap B\toot\sigma\models A\land\sigma\models B$.
  \coqitem[intersect_match_second] Every match of $A\cap B$ is a match of $A$ and $B$.
  \end{enumerate}
\end{fact}
\begin{proof}
  Let $A$ and $B$ be two NFAs 
  with state types $Q_A$ and $Q_B$.
  For $A\cap B$ we use 
  the product $\bool\times Q_A\times Q_B$
  as state type.  A state $(b,p,q)$ is initial
  if $b=\bfalse$, $p$ is an initial state of $A$,
  and $q$ is an initial state of $B$.
  A state $(b,p,q)$ is accepting
  if $b=\btrue$ and~$p$ is an accepting state of $A$.
  The trick is to switch to $\bfalse$
  when leaving an accepting state of $A$, 
  and to $\btrue$ when leaving an accepting state of $B$.
\end{proof}

\begin{fact}[][dec_up_in_lang]
  \label{fact-NFA-UP-sat-dec}
  $\lam{xyA}{~xy^\omega\models A}$ is decidable.
\end{fact}
\begin{proof}[from~\cite{BresolinMP09}]
  Let $x$, $y$ and $A$ be given.
  Fact~\ref{fact-NFA-UP} provides an NFA~$\M{A}_{xy^\omega}$ 
  that accepts exactly the sequences equivalent to $xy^\omega$.
  Now $xy^\omega\models A$ iff $A\cap\M{A}_{xy^\omega}$ is satisfiable.
  The claim follows with Fact~\ref{fact-NFA-solver}.
\end{proof}

\begin{fact}[][match_for_up]
  \label{fact-UP-match}
  If $xy^\omega\models A$, 
  then $A$ has a match $(u,v)$ such that 
  ${xy^\omega\equiv uv^\omega}$.
\end{fact}
\begin{proof}
  Let $xy^\omega\models A$.
  By Fact~\ref{fact-NFA-UP} there is an automaton $\M{A}_{xy^\omega}$ 
  that accepts exactly the sequences equivalent to $xy^\omega$.
  By Fact~\ref{fact-NFA-solver} 
  we obtain a match of $A\cap\M{A}_{xy^\omega}$. 
  The claim follows with Fact~\ref{fact-NFA-inter}.
\end{proof}

\setCoqFilename{Buechi.MinimalS1S}

\begin{fact}
  \label{fact-NFA-exists}
  One can define a function
  that given a variable $X$ and an NFA $A$ over $2^\var$
  yields an NFA \coqlink[ex_nfa]{\emph{$\ex X A$}} over $2^\var$ 
  satisfying the following conditions:
  \begin{enumerate}
  \coqitem[ex_nfa_correct] If $\sigma\models A$, 
    then $\sigma\models\ex XA$.
  \coqitem[ex_nfa_correct] If $\sigma\models\ex XA$,
    then there exists $\tau\models A$ such that $\tau\approx_X\sigma$.
  \coqitem[ex_nfa_up]  If $xy^\omega\models\ex XA$,
    then there exist $u$ and $v$ such that
    $uv^\omega\models A$ and $uv^\omega\approx_X xy^\omega$.
  \end{enumerate}
\end{fact}
\begin{proof}
  Let $X$ and $A$ be given.
  We obtain $\ex XA$ from $A$ by adding transitions:
  for every transition $(p,a,q)$ of $A$ and 
  every set $b:2^\var$ such that $a\cup\set X=b\cup\set X$,
  $\ex XA$ contains the transition $(p,b,q)$.
  The first two claims are easily verified.

  Let $xy^\omega\models\ex XA$.
  Fact~\ref{fact-UP-match} yields 
  a match $(u,v)$ of $\ex XA$ such that
  $xy^\omega\equiv uv^\omega$.
  We can now obtain a match $(w,z)$ of $A$
  such that $uv^\omega\approx_X wz^\omega$.
\end{proof}

\section{NFAs for Ramseyan Factorisations}

Given a finite semigroup $\Gamma$,
we can construct an NFA accepting
exactly the sequences over $\Gamma$ 
that have a Ramseyan factorisation.
The construction exemplifies ideas that
will also appear in the construction
of complement automata.
The result itself will be used for 
the result about the existence of complement automata.

For an NFA we write $\transition{q}{a}{p}$ to say that $(q,a,p)$ is a transition.

\begin{fact}\label{fact-NFA-color}
  For every color $c$ of a finite semigroup $\Gamma$
  one can construct an NFA
  accepting exactly the strings over $\Gamma$
  that have color $c$.
  The NFA can be constructed such that
  it has a single intial state with no 
  incoming transitions.
\end{fact}
\begin{proof}
  Let $\Gamma$ be a finite semigroup and $c:\Gamma$.
  We construct an NFA with $Q:=\one+\Gamma$
  where $\oneM$ serves as initial state and 
  $c$ serves as accepting state.  
  There are transitions $\transition{\oneM}{a}{a}$ and $\transition{b}{a}{b+a}$
  for every $a,b:\Gamma$.
\end{proof}

\begin{fact}\label{fact-NFA-color-omega}
  Given two colors $c$ and $d$ of a finite semigroup $\Gamma$,
  one can construct an NFA accepting all sequences over $\Gamma$
  that have a factorisation~$\tau$ such that $\col(\tau0)=c$
  and $\col(\tau n)=d$ for all $n\ge1$.
\end{fact}
\begin{proof}
  Let $\Gamma$ be a finite semigroup and $c,d:\Gamma$.
  Let~$A$ and~$B$ be the NFAs for~$c$ and~$d$ we obtain
  with Fact~\ref{fact-NFA-color}.
  We start with the disjoint union of $A$ and $B$
  as it comes to states and transitions.
  The new initial state~$q_1$ is the initial state of~$A$,
  and the new accepting state~$q_2$ is the initial state of~$B$.
  For every transition $\transition{q}{a}{p}$ to an accepting state of~$A$
  we add the transition $\transition{q}{a}{q_2}$, and
  for every transition $\transition{q}{a}{p}$ to an accepting state of~$B$
  we add the transition $\transition{q}{a}{q_2}$.
\end{proof}

\begin{fact}\label{fact-NFA-RF}
  For every finite semigroup $\Gamma$ one can construct an NFA
  accepting exactly the sequences over $\Gamma$ 
  that have a Ramseyan factorisation.
\end{fact}
\begin{proof}
  For every pair $(a,b)$ of colors
  we obtain an automaton as specified 
  by Fact~\ref{fact-NFA-color-omega}.
  The union of these automata
  (Fact~\ref{fact-NFA-union})
  satisfies the claimed property.
\end{proof}

\section{Complement Operation}

\setCoqFilename{Buechi.Buechi}

We fix an NFA $A$ over $\Sigma$ with state type $Q$.
We will construct a complement NFA
following Büchi's construction~\cite{buechi1962,thomas99}. 
We carefully arrange the technical details of the construction 
such that the required properties follow constructively and
the connection with Ramseyan factorisations becomes clear.

The letters $x$ and $y$ will range over strings in $\Sigma^+$,
and the letters~$p$ and $q$ will range over states in~$Q$.
Moreover,
the letters $\sigma$, $\tau$, and $\rho$
will range over sequences over $\Sigma$, $\Sigma^+$, and $Q$,
respectively.

We say that \coqlink[accepting_quasi_run]{\emph{$\rho$ is an accepting quasi-run on $\tau$}}
if $\rho$ starts with an initial state,
satisfies ${\trans{\rho n}{\tau n}{\rho(Sn)}}$ for all~$n$,
and satisfies ${\transa{\rho n}{\tau n}{\rho(Sn)}}$
for infinitely many $n$.

\begin{fact}[][accepting_quasi_run_iff_accepts]
  \label{fact-quasi-run}
  Let $\tau$ be a factorisation of $\sigma$.
  Then $A$ accepts $\sigma$ if and only if
  there exists an accepting quasi-run on $\tau$.
\end{fact}

\setCoqFilename{Buechi.Complement}

We define a finite type $\Gamma$ and
a function $\gamma:\Sigma^+\to\Gamma$\,:
\begin{align*}
  \coqlink[Gamma]{\N{\Gamma}}
  &~:=~2^{Q\times Q}\times2^{Q\times Q}
  \\
  \coqlink[gamma]{\N{\gamma x}}
  &~:=~(\,
    \mset{(p,q)}{\trans pxq},
    \mset{(p,q)}{\transa pxq}
    \,)
\end{align*}
Note that $\gamma$ can be defined constructively since
$\trans pxq$ and $\transa pxq$ are decidable predicates
and $Q$ is a finite type.
We call the elements of $\Gamma$ \emph{colors}
and $\gamma x$ the \emph{color of $x$}.
The letters~$V$ and $W$ will range over colors.

\begin{fact}[][same_color_same_accepting_quasi]
  \label{fact-quasi-run2}
  Let $\tau$ and $\tau'$ be sequences over $\Sigma^+$
  such that $\gamma(\tau n)=\gamma(\tau'n)$ for all $n$.
  Then~$\tau$ and $\tau'$ admit 
  the same accepting quasi-runs.
\end{fact}

We now define an operation on colors 
turning $\Gamma$ into a finite semigroup
and $\gamma$ into a semigroup morphism.
\begin{align*}
  \coqlink[addGamma]{\N{V+W}}~:=~(\,
  &\mset{(p,q)}{\exists r.~(p,r)\in\pi_1 V\land(r,q)\in\pi_1 W},\\
  &\mset{(p,q)}{\exists r.~(p,r)\in\pi_2 V\land(r,q)\in\pi_1 W~\lor\\
   & \hspace{63pt}(p,r)\in\pi_1 V\land(r,q)\in\pi_2 W}\,)
\end{align*}

\begin{fact}\label{fact-complement-semigroup}
  \coqlink[addGamma_is_associative]{$V_1+(V_2+V_3)=(V_1+V_2)+V_3$} and
  \coqlink[gamma_is_homomorphisms]{$\gamma(xy)=\gamma x+\gamma y$}.
\end{fact}

A \emph{kind~$V/W$} is a pair of two colors $V$ and $W$.
We say that a sequence $\sigma$ \coqlink[kind]{\emph{has kind~$V/W$}}
if $\sigma$ has a factorisation $\tau$ 
that has color~$V$ at position $0$
and color $W$ at all positions $n\ge1$.
We say that a kind \coqlink[kind_compatible]{\emph{$V/W$ is compatible with $A$}}
if $A$ accepts some sequence of kind $V/W$.
Note that there may be sequences having more than one kind.

\begin{fact}[][compatibility]
  \label{fact-kinds-compatible}
  If\, $V/W$ is compatible with $A$,
  then $A$ accepts all sequences of kind $V/W$.
\end{fact}
\begin{proof}
  Let $\sigma$ and $\sigma'$ be sequences of kind $V/W$
  and let $\sigma\models A$.  We show $\sigma'\models A$.
  Let $\tau$ and $\tau'$ be $V/W$-factorisations 
  of $\sigma$ and $\sigma'$ respectively.
  By Facts~\ref{fact-quasi-run} and~\ref{fact-quasi-run2}
  we have an accepting quasi-run $\rho$ on $\tau$ and $\tau'$.
  Now $\sigma'\models A$  by Fact~\ref{fact-quasi-run}.
\end{proof}

\begin{fact}[][totality_up]
  \label{fact-UP-kind}
  A UP sequence $xy^\omega$ has kind $\gamma x/\gamma y$.
\end{fact}

\begin{fact}[][totality]
  \label{fact-RF-kind}
  Assuming \M{RF}, every sequence over $\Sigma$ has a kind.
\end{fact}
\begin{proof}
  Let $\sigma$ be a sequence over $\Sigma$.
  By \RF{} and Fact~\ref{fact-complement-semigroup}
  we obtain a Ramseyan factorisation $\mu:(\Gamma^+)^\omega$
  of $\lam n{\gamma(\sigma n)}$.
  From $\mu$ we obtain a factorisation 
  $\tau$ of $\sigma$
  such that $\gamma(\tau n)=\col(\mu n)$ for all $n$.
  Thus $\sigma$ has kind $\gamma(\tau0)/\gamma(\tau1)$.
\end{proof}


\begin{fact}[][VW_aut_correct]
  \label{fact-NFA-VW-omega}
  One can define a function
  that given a kind $V/W$ yields 
  an NFA \emph{$VW^\omega$} 
  accepting exactly the sequences of kind $V/W$.
\end{fact}
\begin{proof}
  The construction is similar 
  to the construction given for Fact~\ref{fact-NFA-color-omega}.
\end{proof}


\begin{fact}[][dec_kind_compatible]
  \label{fact-kind-comp-dec}
  It is decidable whether a kind is compatible with $A$.
\end{fact}
\begin{proof}
  $V/W$ is compatible with $A$ 
  if and only if the NFA $VW^\omega\cap A$ is satisfiable.  
  Thus the claim follows with Fact~\ref{fact-NFA-solver}.
\end{proof}

\begin{fact}[][in_complement_nfa_iff]
  \label{fact-complement-automaton}
  One can construct an NFA $\ol A$ 
  accepting exactly the sequences
  that have a kind incompatible with $A$.
\end{fact}
\begin{proof}
  We construct $\ol A$ as the union of the NFAs
  for the kinds incompatible with $A$.
  The construction is possible 
  due to Facts~\ref{fact-NFA-union},
  \ref{fact-kind-comp-dec}, 
  and~\ref{fact-NFA-VW-omega}.
\end{proof}

\begin{theorem}[Complement]
  \label{theo-complement}~
  \begin{enumerate}
  \coqitem[complement_disjoint] No sequence is accepted by both $A$ and $\ol A$.
  \coqitem[complement_kind_exhaustive] Every sequence that has a kind 
    is accepted by either $A$ or $\ol A$.
  \coqitem[complement_exhaustive_up] Every UP sequence is accepted by either $A$ or $\ol A$.
  \coqitem[complement_exhaustive] Assuming \M{RF}, 
    every sequence is accepted by either $A$ or $\ol A$.
  \end{enumerate}
\end{theorem}
\begin{proof}
  Follows with Facts~\ref{fact-complement-automaton},
  \ref{fact-kind-comp-dec}, 
  \ref{fact-kinds-compatible}, 
  \ref{fact-UP-kind},
  and~\ref{fact-RF-kind}.
\end{proof}
 
\section{Main Results so Far}
\label{sec-main-results-so-far}

\setCoqFilename{Buechi.MinimalS1S}

We now have
a translation of \SOSO{} formulas to NFAs
for which we have shown many properties.
We combine the results obtained so far
into main results for \SOSO\ distinguishing
between UP semantics and AS semantics.

\begin{theorem}[\SOSO, AS semantics]
  \label{theorem-as-semantics-s1so}
  Assuming \RF, we have the following:
  \begin{enumerate}
  \coqitem[formula_aut_correct] The translation of formulas to automata 
    is correct for all sequences.
  \coqitem[satisfaction_sat_xm] $\lam{\sigma\phi}{\,\sigma\models\phi}$ satisfies XM. 
  \coqitem[dec_min_S1S_satisfaction]  $\lam\phi{\,\exists\sigma.~\sigma\models\phi}$
    is decidable.
  \end{enumerate}
\end{theorem}
\begin{proof}
  Follows with
  Facts~\ref{fact-NFA-incl}, \ref{fact-NFA-lt},
  \ref{fact-NFA-union}, \ref{fact-NFA-inter}, 
  and~\ref{fact-NFA-exists}, 
  Theorem~\ref{theo-complement},
  and Fact~\ref{fact-NFA-solver}.
\end{proof}

\begin{theorem}[\SOSO, UP semantics]
   \label{theorem-up-semantics-s1so}~
  \begin{enumerate}
  \coqitem[formula_aut_correct] The translation of formulas to NFAs
    is correct for UP sequences.
  \coqitem[dec_satisfaction_up] $\lam{xy\phi}{\,xy^\omega\modelsUP\phi}$ is decidable.
  \coqitem[dec_min_S1S_up_satisfaction] $\lam\phi{\,\exists xy.~xy^\omega\modelsUP\phi}$ is decidable.
  \coqitem[sat_up_obtain] If a formula is UP satisfiable, 
    one can obtain a satisfying UP sequence.
  \coqitem[minsat_agree_as_up] Assuming \RF, $\exists\sigma.\,\sigma\models\phi$ 
    if and only if\,
    $\exists xy.\,xy^\omega\modelsUP\phi$.
  \end{enumerate}
\end{theorem}
\begin{proof}
  Follows with
  Facts~\ref{fact-NFA-incl}, \ref{fact-NFA-lt},
  \ref{fact-NFA-union}, \ref{fact-NFA-inter}, 
  and~\ref{fact-NFA-exists}, 
  Theorem~\ref{theo-complement},
  and Facts~\ref{fact-NFA-solver} and~\ref{fact-NFA-UP-sat-dec}.
\end{proof}

\setCoqFilename{Buechi.NecessityRF}

For the rest of the paper we will be concerned
with results showing that \RF\ is 
a necessary condition for AS semantics.  
We now show that \RF\ is equivalent
to the existence of complement automata.
We also show that \RF\ is equivalent
to the agreement of UP equivalence with
AS equivalence of automata.

Given an NFA $A$ over $\Sigma$,
we call an NFA $A'$ over $\Sigma$
a \emph{complement of $A$} 
if every sequence $\sigma$ over $\Sigma$ satisfies 
(1) $\sigma\models A\to\sigma\models A'\to\bot$
(\textit{disjointness})
and (2)~${\sigma\models A\lor\sigma\models A'}$
(\textit{exhaustiveness}).

Given two NFAs over $\Sigma$,
we write \emph{$A\equiv B$} if $A$ and $B$
accept the same sequences,
and \emph{$A\equivUP B$} if $A$ and $B$
accept the same UP sequences.
We define the following propositions:
\begin{itemize}
\item [\emph{\coqlink[AC]{\AC}}~$:=$]
  For every $\Sigma$,
  every NFA over $\Sigma$ has a complement.
\item [\emph{\coqlink[AU]{\AU}}~$:=$]
  For every $\Sigma$ and
  all $A$ and $B$ over $\Sigma$,  
  $A\equivUP B$ implies $A\equiv B$. 
\item [\emph{\coqlink[AX]{\AX}}~$:=$]
  For every $\Sigma$ and
  all $\sigma$ and $A$ over $\Sigma$,
  either $\sigma\models A$ or $\sigma\nvDash A$. 
\end{itemize}

\begin{fact}[][AC_implies_AX]
  \label{fact-ac-implies-ax}
  \AC\ implies \AX.
\end{fact}

\begin{theorem}[][AC_equiv_RF]
  \label{theorem-rf-equiv-ac}
  \RF\ and  \AC\ are equivalent.
\end{theorem}

\begin{proof}
  $\RF\to\AC$ follows with Theorem~\ref{theo-complement}.
  For the other direction,
  assume \AC\ and let $\Gamma$ be a semigroup.
  By Fact~\ref{fact-NFA-RF}
  we have an NFA~$A$ accepting exactly
  the sequences over $\Gamma$ that have a Ramseyan factorisation.
  Let $A'$ be a complement of $A$ and
  let $\sigma$ be a sequence over~$\Gamma$.
  We show that $\sigma$ has a Ramseyan factorisation.
  Case analysis over $\sigma\models A\lor\sigma\models A'$.
  If $\sigma\models A$, 
  then $\sigma$ has a Ramseyan factorisation by the construction of $A$.
  If $\sigma\models A'$,
  Fact~\ref{fact-NFA-solver} 
  gives us a UP sequence~$xy^\omega$ accepted by~$A'$.
  By disjointness of $A'$ we have $xy^\omega\nvDash A$.
  Contradiction since every UP sequence 
  has a Ramseyan factorisation 
  (Fact~\ref{fact-UP-RF}).
\end{proof}

\begin{theorem}[][AU_equiv_AC]
  \label{theorem-ac-equiv-au}
  \AC\ and \AU\ are equivalent.
\end{theorem}

\begin{proof}

  $\AC\to\AU$ from~\cite{calbrix1993}.
  Assume \AC\ and let $A\equivUP B$.
  Let $\sigma\models A$.  
  By symmetry it suffices to show $\sigma\models B$.
  Let $B'$ be a complement of~$B$. 
  Case analysis using Fact~\ref{fact-NFA-solver}.
  If $A\cap B'$ is unsatisfiable, 
  we have $\sigma\models B$ by exhaustiveness of $B$ and $B'$.
  Otherwise, we have a UP sequence~$xy^\omega$ 
  accepted by $A$ and $B'$.
  Thus $xy^\omega\models B$
  since we assumed $A\equivUP B$. 
  Contradiction with the disjointness of $B$ and $B'$.

  $\AU\to\AC$.
  Assume AU and let $A$ be an NFA.
  We show that~$\ol A$ is a complement for~$A$.
  By Theorem~\ref{theo-complement} 
  we know that~$\ol A$ is disjoint with~$A$.
  Let $A_{\Sigma^\omega}$ be an NFA accepting all sequences.
  By Theorem~\ref{theo-complement} we have
  ${A\cup\ol A\equivUP A_{\Sigma^\omega}}$.
  By the assumption we have ${A\cup\ol A\equiv A_{\Sigma^\omega}}$.
  Thus~$\ol A$ is exhaustive for~$A$.
\end{proof}

\setCoqFilename{Buechi.FullS1S}

\section {Full S1S}
\label{sec-full-s1s}

We now define \emph{S1S} (full S1S) and
give a reduction to \SOSO\ (minimal S1S).
With the reduction, 
the results shown for \SOSO\ carry over to S1S.
One reason for considering S1S 
in addition to \SOSO\ in this paper
is that it better supports the codings
needed for the proof that FX implies RF.

Given two finite types $\var_1$ and $\var_2$ of \emph{variables},  
we obtain the \emph{formulas} of S1S with an inductive type:
\begin{align*}
  \coqlink[Form]{\phi,\psi}~::=~
  &x<y \mid 
    x \in X \mid 
    \phi\land\psi \mid 
    \neg \phi \mid 
    \exists x.\phi \mid 
    \exists X.\phi
  \qquad (x,y : \var_1)~
    (X,Y :\var_2)
\end{align*}
The variables $x$ from $\var_1$ range over numbers
and are called \emph{first-order variables}.
The variables $X$ from $\var_2$ range over
sets of numbers represented as boolean sequences 
and are called \emph{second-order variables}.

Formally, we define the \emph{AS semantics} of S1S
with \emph{interpretations} $I$ 
consisting of two functions 
$\var_1\to\nat$ and $\var_2\to\bool^\omega$.
The \emph{satisfaction relation} \coqlink[satisfies]{$I\models\phi$}
is defined as one would expect.

The reduction of S1S to \SOSO\ 
represents first-order variables
as second-order variables that 
are constrained to singleton sets.
\SOSO\ can express a singleton 
constraint as follows:
\begin{align*}
  \coqlink[Sing]{\sing X}&~:=~\neg(X\lts X)\land\exists Y.\,X\lts Y
\end{align*}
Note that $\neg(X\lts X)$ is satisfied if~$X$
has at most one element,
and that $\exists Y.\,X\lts Y$ is satisfied if~$X$
has at least one element.
If~$X$ and~$Y$ are the singleton variables 
for two first-order variables~$x$ and~$y$,
then $x<y$ can be expressed as $X\lts Y$ and
$x\in Z$ can be expressed as $X\incl Z$.

\begin{fact}\label{fact-s1s-reduction}
  Consider S1S with variable types $\var_1$ and $\var_2$
  and \SOSO\ with variable type $\var_1+\var_2+\one$.
  Then one can obtain for every interpretation~$I$
  and every formula $\phi$ of S1S 
  an interpretation $\hat I$ 
  and a formula $\hat\phi$ of \SOSO\ such that
  \coqlink[full_s1s_to_min_s1s_correct]{$I\models\phi\toot\hat I\models\hat\phi$}.
  Moreover, for every interpretation $\sigma$ of \SOSO\
  that interprets variables from $\var_1$ as singletons,
  one can obtain an interpretation $\tilde\sigma$ of S1S
  such that 
  \coqlink[full_s1s_to_min_s1s_complete]{$\sigma\models\hat\phi\toot\tilde\sigma\models\phi$}.
\end{fact}
\begin{proof}
  We obtain $\hat\phi$ from $\phi$ by constraining
  every variable from~$\var_1$ to a singleton set.
  The \textit{extra variable} provided by $\var_1+\var_2+\one$
  provides for the variable $Y$ in the quantification $\exists Y.\,x\lts Y$
  employed by the singleton constraint.

  We obtain $\hat I$ from $I$ 
  by representing numbers as singleton sets.  
  The value for the extra variable 
  does not matter since the
  extra variable does not occur free in $\hat\phi$.

  Finally, for every interpretation $\sigma$ of \SOSO\
  that interprets variables from $\var_1$ as singletons,
  one can obtain an interpretation $\tilde\sigma$ of S1S
  by assigning to the variables from $\var_1$
  the numbers provided by the singleton sets.
  Obtaining the unique element of
  a singleton set represented as a boolean sequence
  requires constructive choice.
\end{proof}

\begin{theorem}[S1S, AS semantics]
  \label{theorem-as-semantics-s1s}
  Assuming \RF, we have the following:
  \begin{enumerate}
  \coqitem[full_s1s_satisfies_xm] $\lam{I\phi}{\,I\models\phi}$ satisfies XM. 
  \coqitem[full_s1s_dec]  $\lam\phi{\,\exists I.~I\models\phi}$
    is decidable.
  \end{enumerate}
\end{theorem}
\begin{proof}
  Follows with Theorem~\ref{theorem-as-semantics-s1so}
  and Fact~\ref{fact-s1s-reduction}.
\end{proof}

We now formally define the proposition \FX:
$$\coqlink[FX]{\N{\FX}} :=
  \forall\,\var_1\var_2\,I\,\phi.~
  I\models\phi\lor I\nvDash\phi$$
It is understood that in \FX\
the interpretation $I$ and the formula $\phi$
are taken over the types $\var_1$ and $\var_2$, 
which are (notationally implicit) parameters of S1S.

\begin{corollary}[][RF_implies_FX]
  \label{corollary-rf-implies-fx}
  \RF\ implies \FX.
\end{corollary}

The reduction of S1S to \SOSO\ also works for UP semantics.
A \emph{UP interpretation} of S1S consists of two functions 
$\var_1\to\nat$ and $\var_2\to\bool^+\times\bool^+$.
We write \emph{$I\modelsUP\phi$} 
for the satisfaction relation
for UP semantics.

\begin{theorem}[S1S, UP semantics]
   \label{theorem-up-semantics-s1s}~
  \begin{enumerate}
  \coqitem[full_s1s_dec_up_satisfies] $\lam{I\phi}{\,I\modelsUP\phi}$ is decidable.
  \coqitem[full_s1s_dec_up_satisfaction] $\lam\phi{\,\exists I.~I\modelsUP\phi}$ is decidable.
  \coqitem[full_s1s_up_sat_obtain] If a formula is UP satisfiable, 
    one can obtain a satisfying UP interpretation.
  \coqitem[full_s1s_up_sat_if_as_sat] Given \RF, satisfiable formulas are UP satisfiable.
  \end{enumerate}
\end{theorem}
\begin{proof}
  Fact~\ref{fact-s1s-reduction} can be adapted so that
  $\hat I$ is a UP interpretation if I is a UP interpretation,
  and $\tilde\sigma$ is a UP interpretation if $\sigma$
  is a UP interpretation.  
  Now Theorems~\ref{theorem-up-semantics-s1so}
  and~\ref{theorem-as-semantics-s1so}
  yield the claims.
\end{proof}

\setCoqFilename{Buechi.RamseyanPigeonholePrinciple}

\section{Ramseyan Pigeonhole Principle}

We now prepare the proofs of
the implications $\FX\to\RF$ and ${\AX\to\RF}$.
For this, we will define a proposition \RP\
we call \emph{Ramseyan pigeonhole principle}.
We will show that \RP\ is equivalent to~\RF.

We will also consider a variant $\RPc$ of \RP\
that can be obtained from \RP\ by applying
double negation and de Morgan rules.
We will show that $\RPc$ holds constructively.
Now the trick consists in using \FX\ and \AX\
to prove the equivalence $\RP\toot\RPc$.
Since $\RPc$ is obtained from \RP\
by double negation and de Morgan rules,
the special instances of excluded middle
present in \FX\ and \AX\ will suffice
to show the equivalence $\RP\toot\RPc$.
For this it is important that
$\RP$ and $\RPc$ can be expressed with
the satisfaction relations for
formulas and automata.

\RP\ is based on 
a relation for sequences over finite semigroups
appearing as {merging relation} in the 
literature~\cite{buechi1973, mcNaughton66,Perrin04, riba12, Shelah75}.
The relation is used in the literature
to prove Ramseyan properties similar to \RA\ and \RF\
using excluded middle.

Given a sequence~$\sigma$ over a finite semigroup,
we define the \emph{merging relation for~$\sigma$}
as follows:
$$
\coqlink[merge]{\N{\merge{\sigma}{i}{j}}}~:=~
\exists k.~i<k \land j<k \land 
\col (\substr{\sigma}{i}{k}) = \col (\substr{\sigma}{j}{k})
$$
The numbers $i$ and $j$ act as positions of~$\sigma$.
We say that \emph{$i$ merges with~$j$ in $\sigma$}
if $\merge{\sigma}{i}{j}$, 
and that \coqlink[merge_at]{\emph{$i$ merges with $j$ at $k$ in $\sigma$}}
if $i,j<k$ and
$\col(\substr{\sigma}{i}{k})=\col(\substr{\sigma}{j}{k})$.

One easily checks that merging is an equivalence relation using the following fact.
\begin{fact}[][merge_extend]
\label{fact-merging-extend} Let~$\sigma$ be a sequence over a finite semigroup.
If~$i$ merges with~$j$ at~$k$ in~$\sigma$, then~$i$ merges with~$j$ at~$n$ in~$\sigma$ for all $n\geq k$.
\end{fact}

The following fact says that $\merge{\sigma}{i}{j}$
has only finitely many equivalence classes.

\begin{fact}[][merging_finite_equiv_classes]
\label{fact-list-contains-merging-pos}
Let $\sigma$ be a sequence over a finite semigroup $\Gamma$,
$k\ge|\Gamma|$, and
$n_0 < \cdots < n_k$. 
Then there exist numbers $i < j \leq k$ 
with $\merge{\sigma}{n_i}{n_j}$.
\end{fact}
\begin{proof}
Since there are $k+1$ many strings 
$\substr{\sigma}{n_0}{n_k +1},\dots,\substr{\sigma}{n_k}{n_k + 1}$
and at most $k$ many colors, 
there are two numbers $i < j \leq k$ such that
$\col(\substr{\sigma}{n_i}{n_k+1})=\col (\substr{\sigma}{n_j}{n_k+1})$.
\end{proof}

We now define the propositions $\RP$ and $\RPc$:
\begin{align*}
  \coqlink[RP_sigma]{\N{\RPs{\sigma}}}  
  &~:=~\exists i\,\exists^\omega j.~
    \merge{\sigma}{i}{j} 
  &\coqlink[RP]{\N{\RP}}         
  &~:=~\forall \Gamma\,\forall \sigma{:}\,\Gamma^\omega.~
    \RPs{\sigma}
  \\
  \coqlink[RPc_sigma]{\N{\RPsc{\sigma}}}
  &~:=~\neg \forall i\,\exists k\,\forall j \geq k.~
    \notmerge{\sigma}{i}{j}
  &\coqlink[RPc]{\N{\RPc }}         
  &~:=~\forall \Gamma\,\forall \sigma{:}\,\Gamma^\omega.~
    \RPsc{\sigma}
\end{align*}
Note that \RP\ states that 
every sequence over a finite semigroup 
has an infinite merging class.

\begin{fact}
\label{fact-xm-implies-rp-equiv-rpc}
Assuming \XM{}, \RP{} is equivalent to $\RPc{}$.
\end{fact}
\begin{proof} By application of de Morgan laws and double negation.
\end{proof}
\begin{fact}[][RPc_holds]
\label{fact-rpc-holds}
$\RPc{}$ holds.
\end{fact}
\begin{proof}
Let $\sigma$ be a sequence over a finite semigroup $\Gamma$.
We assume 
${\forall i\,\exists k\,\forall j\geq k.~\notmerge{\sigma}{i}{j}}$ 
and derive a contradiction.
We show by induction on $k$ that for every $k$ 
there are pairwise non-merging  numbers $n_0<\cdots<n_k$
such that 
$\forall i{<}k\,\forall j{\ge}n_k.~\notmerge{\sigma}{n_i}{j}$.
Contradiction with Fact~\ref{fact-list-contains-merging-pos}.
\end{proof}

\begin{fact}[][RF_implies_RP]
\label{fact-rf-implies-rp}
\RF{} implies \RP{}.
\end{fact}
\begin{proof}
Assume \RF{} and let $\sigma$ be a sequence over a finite semigroup. 
We show $\RPs\sigma$.
By  Fact~\ref{fact-RF-inf-homo}
there is an infinite boolean sequence $\beta$ such that 
$\sgsum (\substr{\sigma}{i}{j})$ is constant 
for all $i,j \in \beta$ with $i < j$.
We fix some $i\in\beta$ and some number $k$.
Then there are $j,l\in\beta$ such that 
$k\le j$ and $i,j<l$.
We now have $\merge{\sigma}{i}{j}$ since
$\col(\substr{\sigma}{i}{l})=\col(\substr{\sigma}{j}{l})$.
\end{proof}


\begin{fact}[][RP_implies_IP]
\label{fact-rp-implies-ip}
\RP{} implies \IP{}.
\end{fact}
\begin{proof}
Assume \RP{} and 
let $\sigma$ be a sequence over a finite type~$\Sigma$.
We use $\lambda a b. a$ as semigroup operation for~$\Sigma$.
By $\RPs{\sigma}$ there is a number~$i$ 
and infinitely many numbers~$j$ merging with~$i$.
Let~$i$ merge with~$j$ at~$k$.
Then 
$\sigma j = 
\sgsum (\substr{\sigma}{j}{k}) = 
\sgsum (\substr{\sigma}{i}{k}) = 
\sigma i$.
Hence~$\sigma i$ occurs infinitely often in~$\sigma$.
\end{proof}

\begin{fact}[][RP_implies_RF]
\label{fact-rp-implies-rf}
\RP{} implies \RF{}.
\end{fact}
\begin{proof}
  Assume \RP\ and let
  $\sigma$ be a sequence over a finite semigroup.
  We show that $\sigma$ has a Ramseyan factorisation
  using Fact~\ref{fact-RF-inf-homo}.

  By $\RPs\sigma$ we have a number $m$
  such that $\exists^\omega i.\,\merge\sigma mi$.
  Using constructive choice,
  we obtain a function $f:\nat\to\nat\times\nat$
  such that $fn$ yields two numbers $n<i<k$ such that
  $m$ merges with $i$ at $k$ in~$\sigma$
  (possible since 
  $\lam k{\,\exists i<k.~n<i\land
    \col(\substr\sigma mk)=\col(\substr\sigma ik)}$ 
  is decidable).
  
  We now obtain a strictly monotone function $g:\nat\to\nat$
  such that $m$ merges with $gi$ at $gj$ in~$\sigma$
  for all $i<j$
  (using Fact~\ref{fact-merging-extend}).

  We consider the sequence $\lam n{\col(\substr\sigma m{gn})}$.
  By \IP\ (Fact~\ref{fact-rp-implies-ip}) 
  there exist a color $a$ such that 
  $\exists^\omega n.\, \col(\substr\sigma m{gn})=a$.
  We now obtain (using Fact~\ref{fact-bseq-enumeration})
  an infinite boolean sequence $\beta$
  such that for all $n$
  $$
  n\in\beta~\toot~
  \exists k.~n=gk\land \col(\substr\sigma m{gk})=a
  $$
  Let $i,j\in\beta$ such that $i<j$.
  By Fact~\ref{fact-RF-inf-homo} 
  it suffices to show $\col(\substr\sigma ij)=a$.
  By the definition of $\beta$ 
  we have $k_i<k_j$ such that $i=g(k_i)$, $j=g(k_j)$,
  and $\col(\substr\sigma m{g(k_j)})=a$.
  Now $\col(\substr\sigma ij)=
  \col(\substr\sigma {g(k_i)}{g(k_j)})=
  \col(\substr\sigma {m}{g(k_j)})=a$.
\end{proof}

The proof of Fact~\ref{fact-rp-implies-rf}
is complicated by the fact that $\beta$
must be obtained as a computational function
in constructive type theory.
Similar constructions 
carried out in classical set theory appear 
in the literature~\cite{buechi1973, Perrin04, riba12, Shelah75}
as part of proofs of properties similar to \RF.
The properties $\RP$ and $\RPc$ 
are not made explicit in the literature.
Recall that we have given 
constructive proofs of $\RPc$ and $\RP\to\RF$,
and that there is a trivial classical proof of $\RP\toot\RPc$.

\begin{corollary}
  \XM\ implies \RF.
\end{corollary}

Note that \RA\ and \RF\ 
existentially quantify over functions 
and that this is not the case for \RP.
The proof of $\RP\to\RF$ reveals 
how in a constructive setting one
can construct a nontrivial function
from existential assumptions for numbers.

\setCoqFilename{Buechi.NecessityRF}

\section{FX implies RP}

To show that \FX\ implies \RP, 
we encode \RP{} and $\RPc{}$ into S1S and 
show that the encodings are equivalent. 
Recall from Section~\ref{sec-full-s1s} 
that \FX{} says that satisfaction in S1S satisfies XM. 
Assuming \FX{}, double negation and de Morgan laws hold in S1S 
and we can use universal quantification and all boolean connectives.

We choose a finite type $\var_1$ 
providing at least three distinct first-order variables.

\begin{fact}[][phi_merge_correct]
\label{fact-s1s-formula-for-merging}
Assume \FX{} and let $\sigma$ be a sequence over a finite semigroup $\Gamma$ and $x,y : \var_1$ be two distinct first-order variables.
There is an interpretation $I_\sigma$ and a formula $\varphi_{x \sim y}$ with variable types $\var_1$ and $\var_2 := \Gamma + \Gamma$
such that $\fullsat{I_\sigma[x:=i,y:=j]}{\varphi_{x \sim y}}$ if and only if  $\merge{\sigma}{i}{j}$ for all $i$ and $j$.
\end{fact}
\begin{proof}
The formula is given in Figure~\ref{fig-encoding-rp-in-s1s}. The sequence $\sigma$ is encoded as usual \cite{Perrin04}: There are free second-order variables $X_a$ for all $a : \Gamma$ and the interpretation $I_\sigma$ is defined such that $I_\sigma X_a$ is the boolean sequence containing exactly the positions at which $\sigma$ is $a$.

We provide informal explanations for the formulas in Figure~\ref{fig-encoding-rp-in-s1s} (for readability, we use S1S variables in equations):
 \begin{itemize}
 \item $\varphi_{\sgsum_x^y=c}$ says that \coqlink[phi_sum_eq_correct]{the color of $\substr{\sigma}{x}{y}$ is $c$}. The variables $Y_a$ encode the colors of $\substr{\sigma}{x} {Sx}, \dots, \substr{\sigma}{x}{y}$.
 \item $\varphi_{\mathsf{unique}}$ says that \coqlink[phi_unique_correct]{$z$ can be in at most one $Y_a$}.
 \item $\varphi_{\mathsf{first}}$ says that  \coqlink[phi_first_correct]{$\sgsum (\substr{\sigma}{x}{Sx}) = \sigma x$}.
 \item $\varphi_{\mathsf{step}}$ says that  \coqlink[phi_step_correct]{$\sgsum (\substr{\sigma}{x}{Sz}) = \sgsum(\substr{\sigma}{x}{z}) + \sigma z$ for $x < z <y$}.
 \item $\varphi_{\mathsf{last}}$ says that \coqlink[phi_last_correct]{the color of $\substr{\sigma}{x}{y}$ is $c$}.
 \end{itemize}
%
\vskip-4mm
\end{proof}

\begin{figure}
  \begin{align*}
    \coqlink[phi_merge]{\varphi_{x \sim y}}   
    &:= \exists z.~x < z \land y < z \land 
      \bigvee_{c : \Gamma} \left( 
      \varphi_{\sgsum_x^z=c} \land   \varphi_{\sgsum_y^z=c} 
      \right)
    \\
    \coqlink[phi_sum_eq]{\varphi_{\sgsum_x^y=c}}
    &:= \underset{a : \Gamma}{\raisebox{-0.5ex}{\makebox{\LARGE$\exists$}}} Y_a.~
      \varphi_{\mathsf{unique}} \land \varphi_{\mathsf{first}} \land
      \varphi_{\mathsf{step}} \land \varphi_{\mathsf{last}}
    \\
    \coqlink[phi_unique]{\varphi_{\mathsf{unique}}}
    &:= \forall z.~\bigand_{a : \Gamma} \Big(z \in Y_a \rightarrow \bigand_{b \neq a} z \notin Y_b \Big)
    \\
    \coqlink[phi_first]{\varphi_{\mathsf{first}}}  
    &:= \bigand_{a : \Gamma} \left(x \in X_a \rightarrow S x \in Y_a \right)
    \\
    \coqlink[phi_step]{\varphi_{\mathsf{step}}} 
    &:= \forall z.~x < z <y  \rightarrow  \bigand_{a,b : \Gamma} \left( z \in Y_a \to z \in X_b \to  S z \in Y_{a + b} \right)
    \\
    \coqlink[phi_last]{\varphi_{\mathsf{last}}}   
    &:= y \in Y_c
\end{align*}
\caption{Encoding of $\merge{\sigma}{i}{j}$ into S1S for Fact~\ref{fact-s1s-formula-for-merging}. 
We write $\varphi(S x)$ for $\forall x'.~ x < x' \to (\neg \exists y. ~x < y < x') \to \varphi(x')$.}
\label{fig-encoding-rp-in-s1s}
\end{figure}

\begin{fact}[][FX_implies_RP]
  \label{fact-fx-implies-rp}
  \FX{} implies \RP{}.
\end{fact}
\begin{proof}
  Assume \FX\ and
  let $\sigma$ be a sequence over a finite semigroup. 
  By Fact~\ref{fact-rpc-holds}
  it suffices to show that $\RPsc\sigma$ entails $\RPs\sigma$.
  We encode $\RPsc\sigma$ and $\RPs\sigma$ in S1S using
  Fact~\ref{fact-s1s-formula-for-merging}:
  \begin{align*}
    \RPs{\sigma} 
    &\toot \fullsat{I_\sigma}{\exists x. ~\forall z. ~\exists y\geq z.~\varphi_{x \sim y}
    \\ 
    \RPsc{\sigma}} 
    &\toot \fullsat{I_\sigma}{\neg \forall x.~\exists z. ~\forall y \geq z.~
      \neg \varphi_{x \sim y}}
  \end{align*}
  Both encodings can be shown equivalent in S1S 
  using de Morgan laws and double negation.
  Thus $\RPsc\sigma$ entails $\RPs\sigma$.
\end{proof}

Siefkes~\cite{siefkes70} and Riba~\cite{riba12} 
show in a classical setting that
S1S can encode propositions similar to $\RA$ and $\RP$,
respectively.
In contrast to Riba~\cite{riba12},
we use an explicit encoding of 
propositions
$\sgsum(\substr{\sigma}{i}{j})=a$.

\section{AX implies RP}

We finally show that $\AX$ implies $\RP$.
Recall from Section~\ref{sec-main-results-so-far} 
that $\AX$ says that acceptance by automata satisfies XM.
Assuming $\AX$, we will show that $\RP$ and $\RPc$ are equivalent.
There is the difficulty that we cannot encode 
$\RP$ and $\RPc$ into automata since
this requires complement automata,
which we do not have since we do not have~$\RF$.
We solve the problem with three predicates
that can be encoded into automata without using complement
and that suffice to justify the uses of double negation
needed for the equivalence proof.

We define the helper predicates as follows,
where $\Gamma$ ranges over finite semigroups 
and $\sigma$ over sequences over $\Gamma$: 
\begin{align*}
  \coqlink[P1]{\N{\pred{1}}}
  &~:=~\lambda \Gamma \sigma i k.~
    \exists\, j \geq k.~ 
    \merge{\sigma} i j
  \\
  \coqlink[P2]{\N{\pred{2}}} 
  &~:=~\lambda \Gamma \sigma.~
    \exists i\,\exists^\omega j.~ 
    \merge{\sigma}{i}{j}
  \\
  \coqlink[P3]{\N{\pred{3}}}
  &~:=~\lambda \Gamma \sigma i.~
    \exists k.\,\neg\exists\, j \geq k.~ 
    \merge{\sigma}{i}{j}
\end{align*}
Note that $\pred{2}$ is $\lambda \Gamma \sigma. \RPs{\sigma}$.

\begin{fact}[][pred_xm_RP]
  \label{fact-xm-implies-rp-equiv-refine}
  Let \pred1, \pred2,  and \pred3 satisfy XM.
  Then $\RP$ holds.
\end{fact}
\begin{proof}
  Let $\sigma$ be a sequence 
  over a finite semigroup $\Gamma$.
  By Fact~\ref{fact-rpc-holds} it suffices to show
  that $\RPs\sigma$ and $\RPsc\sigma$ are equivalent. 
  This follows with the assumptions 
  and Fact~\ref{fact-xm-rules}:
  \begin{align*}
    &\exists i\,\forall k\,\exists j.~
      j\ge k\land\merge{\sigma}{i}{j}
    \\
    \toot~&\neg\forall i.~\neg\forall k.~\neg\neg\exists j.~
      j\ge k\land\merge{\sigma}{i}{j}
    &&\pred2,~\pred1
    \\
    \toot~&\neg\forall i\,\exists k.~\neg\exists j.~
      j\ge k\land\merge{\sigma}{i}{j}
    &&\pred3 
    \\
    \toot~&\neg\forall i\,\exists k\,\forall j.~
      j\ge k\to\notmerge{\sigma}{i}{j}
  \end{align*}
  \vskip-5mm
\end{proof}

We now encode the predicates $\pred i$
into automata and show that they satisfy XM
if \AX\ is assumed.

\begin{fact}[][sat_xm_nat_dec_exists]
  \label{fact-ax-implies-ex-pn-xm}
  Assume \AX\ and
  let $p$ be a decidable predicate on numbers.
  Then $\exists n. p n$ satisfies XM.
\end{fact}
\begin{proof}
Let $A$ be an automaton that accepts exactly all nonempty boolean sequences.
We define $\beta$ to be a boolean sequence
satisfying $n \in \beta \toot p n$.
Note that $\beta$ can be defined
because $p$ is decidable.
Then $\exists n. p n$ is equivalent to $\beta \models A$
and satisfies XM by \AX{}.
\end{proof}

\begin{fact}[][sat_xm_exists_merging]
  \label{fact-ax-implies-p1-xm}
  Assume \AX{}.
  Then \pred{1} satisfies XM.
\end{fact}
\begin{proof}
  First note that $\pred1\Gamma\sigma ik$ is equivalent to
  $$
  \exists l.~i<l\land\exists j.~j<l\land j\ge k\land
  \col(\substr{\sigma}{i}{l}) = \col (\substr{\sigma}{j}{l})
  $$
  Now the claim follows with Fact~\ref{fact-ax-implies-ex-pn-xm}
  since the quantification over~$j$ is bounded and thus decidable.
\end{proof}

We denote with \emph{$\dropi\sigma$} 
the sequence obtained from $\sigma$ 
by dropping the first $i$ positions.

\setCoqFilename{Buechi.RamseyanPigeonholePrinciple}
\begin{fact}[][merging_shift]
\label{fact-merging-shift}
Let $\Gamma$ be a finite semigroup, $\sigma$ be a sequence over $\Gamma$, and $i\le j$. 
Then $\merge{\sigma}{i}{j}~\toot~\merge{\dropi\sigma}{0}{(j-i)}$.
\end{fact}

\setCoqFilename{Buechi.NecessityRF}

\begin{fact}[][aut_exists_inf_merging_correct]
\label{fact-nfa-c}
Let $\Gamma$ be a finite semigroup. 
Then there is an NFA $A$ such that
${\sigma \models A \toot \exists^\omega j. ~\merge{\sigma}{0}{j}}$
for all sequences $\sigma$ over $\Gamma$.
\end{fact}
\begin{proof}
The NFA $A$ repeatedly guesses a position and verifies that it merges with $0$.
The state type of $A$ is $(\one + \Gamma) \times (\one + \Gamma)$.
The NFA $A$ computes the color of the processed prefix  $\substr{\sigma}{0}{n}$ in the first component.
After guessing a position $j$, $A$
computes the color of~$\substr{\sigma}{j}{n}$ in the second component.
Once these two colors are equal,
the guess was correct.

In the first component, $\oneM$
is only used for the initial state.
In the second, $\oneM$
encodes the phase while $A$ is guessing the next position.
Hence $(\oneM, \oneM)$ serves as initial state.
The accepting states are $(a,a)$ for all $a:\Gamma$, 
meaning that $A$ verified a guess.
For all $a$ and $b$ there are transitions 
$\transition{(\oneM, \oneM)}{a}{(a, \oneM)}$,
$\transition{(b, \oneM)}{a}{(b + a, \oneM)}$, and
$\transition{(b, \oneM)}{a}{(b + a, a)}$ to guess a position.
To verify the guess there are for all $a$ and $b \neq c$ transitions
$\transition{(b, c)}{a}{(b + a, c + a)}$ and
$\transition{(b, b)}{a},{(b + a, \oneM)}$. 

If a run of $A$ passes infinitely often through accepting states, $A$ guessed infinitely many positions merging with $0$.
Conversely, if there are infinitely many positions merging with $0$, then $A$ accepts~$\sigma$.
If $0$ merges with $j$ at position $k$, there is always another $j'> k$ merging with $0$, which $A$ can guess.
\end{proof}

\begin{fact}[][aut_inf_merging_correct]
\label{fact-nfa-p2}
Let $\Gamma$ be a finite semigroup.
Then there is an NFA $A$ such that
${\sigma \models A \leftrightarrow 
 \exists i\, \exists j^\omega.~ \merge{\dropi\sigma}{0}{j}}$
 for all sequences~$\sigma$ over $\Gamma$. 
\end{fact}
\begin{proof}
Let~$B$ be the NFA from Fact~\ref{fact-nfa-c} for~$\Gamma$.
Then~$A$ is obtained from~$B$ as follows:~$A$ guesses~$i$
by reading the first $i$ positions of~$\sigma$ and
then transitions into the initial state of~$B$.
\end{proof}

\begin{fact}[][sat_xm_exists_index_inf_merging]
  \label{fact-ax-implies-p2-xm}
  Assume \AX{}. Then \pred{2} satisfies XM.
\end{fact}
\begin{proof}
  Let~$\sigma$ be a sequence over a finite semigroup $\Gamma$.
  By Fact~\ref{fact-nfa-p2} it suffices to show that
  $\pred{2}\Gamma \sigma$ is equivalent to 
  $\exists i\, \exists j^\omega.~\merge{\dropi\sigma}{0}{j}$:
  \begin{align*}
    &\exists i\, \exists^\omega j.~
      \merge{\sigma}{i}{j}\\
    \toot~&\exists i\, \exists^\omega j.~
      j \geq i \land ~\merge{\sigma}{i}{j} \\
    \toot~&\exists i\, \exists j^\omega.~
      j \geq i \land \merge{\dropi\sigma}{0}{(j-i)}
      && \text{Fact~\ref{fact-merging-shift}} \\
    \toot~&\exists i\, \exists j^\omega.~
      \merge{\dropi\sigma}{0}{j}
  \end{align*}
  \vskip-5mm
\end{proof}

\begin{fact}[][aut_forall_sums_different]
\label{fact-nfa-p3}
Let $\Gamma$ be a finite semigroup.
Then there is an NFA $A$ such that
${\sigma \models A \toot \exists k.~\forall j \geq k.~\notmerge{\sigma}{0}{j}}$
for all sequences $\sigma$ over $\Gamma$.
\end{fact}
\begin{proof}
The automaton~$A$ first guesses~$k$ and
then asserts that all greater positions
do not merge with~$k$.
The state type is ${(\one + \Gamma) \times (\one + 2^\Gamma)}$.
The first component is used to compute the
color of $\substr{\sigma}{0}{n}$ for $n > 0$ and
the second component to compute the 
set ${\mset{\sgsum (\substr{\sigma}{j}{n})}{k \leq j < n}}$
of colors of all suffixes of $\substr{\sigma}{k}{n}$.
The initial state is $(\oneM, \oneM)$.

To verify that the guessed position~$k$ was
correct,~$A$ needs to ensure that the color
of $\substr{\sigma}{0}{n}$ is never equal to
the color of a suffix $\substr{\sigma}{j}{n}$
with $k\leq j < n$ (as then $j$ merges with $0$).
If that is the case,~$A$ gets stuck and cannot
continue running. Hence all states $(a, s)$
are accepting for $a:\Gamma$ and $s:2^\Gamma$. 

There are transitions
${\transition{(\oneM, \oneM)}{a}{(a, \oneM)}}$,
${\transition{(b, \oneM)}{a}{(b + a, \oneM)}}$, and
${\transition{(b, \oneM)}{a}{(b+a,\set{a})}}$
for all ${a:\Gamma}$ to guess $k$.
To verify the guess there are
for all $a, b, c : \Gamma$ and $s: 2^\Gamma$ with $b \notin s$
transitions
${\transition{(b,s)}{a}{(b + a, \set{a}\cup\mset{c + a }{ c \in s})}}$.
\end{proof}

\begin{fact}[][sat_xm_exists_not_merging]
  \label{fact-ax-implies-p3-xm}
  Assume \AX{}.
  Then \pred{3} satisfies XM.
\end{fact}
\begin{proof}
  Let~$\sigma$ be a sequence over a finite semigroup $\Gamma$
  and $i$ be a number.
  By Fact~\ref{fact-nfa-p3} it suffices to show that
  $\pred{3}\Gamma \sigma i$ is equivalent to 
  $\exists k \,\forall j \geq k.~\notmerge{\dropi\sigma}{0}{j}$:
  \begin{align*}
     &\exists k.~\neg \exists j \geq k.~\merge{\sigma}{i}{j}\\
     \toot~&\exists k \,\forall j \geq k.
         ~\notmerge{\sigma}{i}{j}
         && \text{Fact~\ref{fact-xm-rules}}\\
     \toot~&\exists k \, \forall j \geq (k + i).
         ~\notmerge{\sigma}{i}{j}\\
     \toot~&\exists k \,\forall j \geq (k + i).
         ~\notmerge{\dropi\sigma}{0}{(j-i)}
         && \text{Fact~\ref{fact-merging-shift}}\\
     \toot~&\exists k \,\forall j \geq k.
         ~\notmerge{\dropi\sigma}{0}{j}
  \end{align*}
  \vskip -5mm
\end{proof}

\begin{fact}[][AX_implies_RP]
\label{fact-ax-implies-rp}
\AX{} implies \RP{}.
\end{fact}
\begin{proof}
Follows with Facts~\ref{fact-xm-implies-rp-equiv-refine},
\ref{fact-ax-implies-p1-xm}, 
\ref{fact-ax-implies-p2-xm}, and
\ref{fact-ax-implies-p3-xm}.
\end{proof}

\begin{theorem}[][main_all_equiv]
\label{theorem-all-equivalences}
\FX{}, \AX{}, \AC{}, \AU{}, \RF{}, \RA{}, and \RP{} are pairwise equivalent.
\end{theorem}
\begin{proof}
Follows with
Facts \ref{fact-ra-implies-rf},
\ref{fact-rf-implies-ra},
and~\ref{fact-ac-implies-ax},
Theorems~\ref{theorem-rf-equiv-ac} and~\ref{theorem-ac-equiv-au},
Corollary~\ref{corollary-rf-implies-fx},
and Facts~\ref{fact-rp-implies-rf},
\ref{fact-fx-implies-rp},
and~\ref{fact-ax-implies-rp}.
\end{proof}

Note that each of the propositions 
in Theorem~\ref{theorem-all-equivalences} 
follows from $\XM$ 
(since $\FX$ follows from $\XM$)
and is unprovable constructively 
(since $\RF$ implies $\IP$ and $\MP$, Fact~\ref{fact-IP-MP}).

\section{Final Remarks}

In this paper we have studied the reduction
of S1S to Büchi automata in Coq's constructive type theory.
We have worked with two different semantics,
AS semantics and UP semantics.
For UP semantics, we showed without assumptions
that Büchi's complement operation is correct
and that S1S is decidable and classical.
For AS semantics, we obtained these results
assuming RF (a weak version of Ramsey's theorem
following with excluded middle).
We showed that the assumption RF is 
strictly necessary for AS semantics
since (1) it is constructively unprovable
and (2) it is entailed by
each of the following properties:
Complement automata exist, 
automaton acceptance satisfies XM, and
formula satisfaction satisfies XM.

AS semantics is the canonical semantics 
for S1S and Büchi automata in the literature.
Our results show that AS semantics
does not work constructively.
To make it work we need to assume RF.
While RF is a consequence of excluded middle,
it seems unlikely that RF entails excluded middle.

UP semantics admits only 
ultimately periodic sequences $xy^\omega$
specified by two strings $x$ and $y$.
It is not surprising that UP semantics
works constructively, and that 
UP semantics agrees with AS semantics 
if RF is assumed.

Doczkal and Smolka~\cite{DoczkalSmolka2014}
give a purely constructive development
of the temporal logic CTL.  
To make this possible, they admit only
finite transition systems as models.
Using tableau methods,
they prove decidability and 
show soundness and completeness
of a standard Hilbert proof system.
This way they establish 
in a purely constructive way that
the constructive semantics agrees
with the standard semantics of CTL
as given by the Hilbert system.
Assuming XM and dependent choice, 
they also show that 
the standard path semantics of CTL
agrees with the constructive semantics.

Sound and complete proof systems
for S1S exist~\cite{riba12,Shelah75,siefkes70}.
We expect that 
soundness and completeness for UP semantics
can be shown constructively.
For AS semantics, 
RF will be necessary for soundness.

In this paper, we have only considered
Büchi's complement operation~\cite{buechi1962,thomas99}.
We expect that other complement operations,
in particular complementation by transformation
to deterministic Muller automata~\cite{mcNaughton66,thomas99},
can also be verified for AS semantics given RF.
Recall that we have shown that 
RF is needed for the verification
of every complement operation.

As it comes to future work,
we plan to extend the constructive analysis 
of automata-based decision methods 
from S1S to further logics such as 
LTL, CTL, and S2S.

\bibliography{s1s}
\bibliographystyle{abbrv}

\end{document}